\documentclass[twocolumn]{article}

\usepackage{amsmath,graphicx}
\usepackage[latin1]{inputenc}
\usepackage{amssymb}  
\usepackage{amsthm}
\usepackage{amsfonts}
\usepackage{dsfont}
\usepackage{algorithm, algorithmic}
\usepackage{graphicx}
\usepackage{color}
\usepackage{placeins}
\usepackage{float}
\usepackage{tabularx,colortbl}

\newcommand{\R}{\mathbb{R}}

\newcommand{\V}{\mathcal{V}}
\newcommand{\C}{\mathcal{C}}
\newcommand{\E}{\mathcal{E}}
\newcommand{\G}{\mathcal{G}}
\newcommand{\T}{\mathcal{T}}

\newcommand{\e}{\mbox{e}}

\newtheorem{theorem}{Theorem}
\newtheorem{corollary}{Corollary}

\usepackage{etoolbox}
\newtoggle{showrevisions}

\toggletrue{showrevisions}

\iftoggle{showrevisions}{%
\usepackage{xcolor}
\usepackage{changebar}
\usepackage[normalem]{ulem}

\newcommand{\removed}[1]{\cbstart\removedfragile{#1}\cbend{}}
\newcommand{\removedfragile}[1]{{\color{red}{\sout{#1}}}{}}

}{%
  \newcommand{\removed}[1]{} 
  \newcommand{\removedfragile}[1]{}

}

\begin{document}

\title{Steerable Discrete Fourier Transform}

\author{Giulia Fracastoro, Enrico Magli\footnote{The authors are with the Department
of Electronics and Telecommunications, Politecnico di Torino, Torino, Italy (e-mail: name.surname@polito.it).
This work has been supported by Sisvel Technology.}}
        

\maketitle

\begin{abstract}
Directional transforms have recently raised a lot of interest thanks to their numerous applications in signal compression and analysis. In this letter, we introduce a generalization of the discrete Fourier transform, called steerable DFT (SDFT). Since the DFT is used in numerous fields, it may be of interest in a wide range of applications. Moreover, we also show that the SDFT is highly related to other well-known transforms, such as the Fourier sine and cosine transforms and the Hilbert transforms. 
\end{abstract}



\section{Introduction}

In the last few years, several authors have proposed using directional transforms for various signal and image processing tasks. Examples include the directional \cite{zeng2008directional} and steerable \cite{fracastoro2015steerable} discrete cosine transform, the rotational transform \cite{alshina2011rotational}, as well as other transforms employing sophisticated nonseparable geometries, e.g. curvelets \cite{candes2000curvelets}, bandlets \cite{le2005sparse}, contourlets \cite{do2005contourlet}, and so on. Such transforms are appealing in many applications, including signal analysis and compression, because the adaptation of geometric parameters can optimally match the transform to the signal of interest. 

Along the same lines, the discrete Fourier transform (DFT) is one of the most important tools in digital signal processing. It enables us to analyze, manipulate, and synthesize signals and it is now used in almost every field of engineering \cite{lyons2010understanding}. In the past, some generalizations of the Fourier transform have been presented, such as the short time Fourier transform \cite{allen1977unified} and the fractional Fourier transform (also called angular Fourier transform) \cite{almeida1994fractional} \cite{almeida1993introduction}. The short time Fourier transform subdivides the signal into narrow time intervals in order to obtain simultaneous information on time and frequency. Instead, the fractional Fourier transform, and its discrete version called discrete rotational Fourier transform \cite{santhanam1996discrete}, can be interpreted as a rotation on the time-frequency plane. Recently, the concept of a graph Fourier transform (GFT) has been introduced in \cite{hammond2011wavelets}; this new transform generalizes the traditional Fourier analysis to the graph domain.

In \cite{fracastoro2015steerable}, the theory of graph signal processing \cite{shuman2013emerging}, and particularly the relationship between the graph Fourier transform and grid graphs, has been exploited to define a new directional 2D-DCT \cite{strang1999discrete} that can be steered in a chosen direction. 
In this letter, we extend this concept and present a new generalization of the DFT, called steerable discrete Fourier transform (SDFT). The proposed SDFT can be defined in one or two dimensions (unlike the steerable DCT which can be defined only in the 2D case). In 1D, we start from the definition of the GFT of a cycle graph and we obtain a new transform, the 1D-SDFT, by rotating the 1D-DFT basis. The 1D-SDFT can be interpreted as a rotation of the basis vectors on the complex plane. 
Instead, in the 2D case we use the GFT of a toroidal grid graph to introduce the new 2D-SDFT, which can be obtained by rotating the 2D-DFT basis. The 2D-SDFT represents a rotation on the two-dimensional Euclidean space.
Since the DFT is used in a wide range of applications, the SDFT represents an interesting generalization that could be applied in various fields, including e.g. filtering, signal analysis, or even multimedia encryption where parametrized versions of common transforms have been used for security purposes \cite{pande2013secure, unnikrishnan2000optical}. We also show that the SDFT is related to other well-known transforms, such as the Fourier sine and cosine transforms and the Hilbert transform.

\section{Basic definitions on graphs}
A graph can be denoted as $\G=(\V,\E)$, where $\V$ is the set of vertices (or nodes) with $|\V|=N$ and $\E\subset \V\times \V$ is the set of edges. It is possible to represent a graph by its adjacency matrix $A(\G)\in \mathbb{R}^{N\times N}$, where $A(\G)_{ij}=1$ if there is an edge between node $i$ and $j$, otherwise $A(\G)_{ij}=0$. The graph Laplacian is defined as $L(\G)=D(\G)-A(\G)$, where $D(\G)$ is a diagonal matrix whose $i$-th diagonal element $D(\G)_i$ is equal to the number of edges incident to node $i$. Since $L(\G)$ is a real symmetric matrix, it is diagonalizable by an orthogonal matrix $L(\G)=\Phi \Lambda\Phi^H$, where $\Phi\in\mathbb{R}^{N\times N}$ is the eigenvector matrix of $L(\G)$ that contains the eigenvectors as columns, $\Lambda$ is the diagonal eigenvalue matrix where the eigenvalues are sorted in increasing order and $H$ denotes the Hermitian transpose.

A graph signal $\mathbf{x}\in\mathbb{R}^N$ in the vertex domain is a real-valued function defined on the nodes of the graph $\G$ such that $\mathbf{x}_i$, where $i=1,...,N$, is the value of the signal at node $i\in\V$ \cite{shuman2013emerging}. The eigenvectors of $L(\G)$ are used to define the graph Fourier transform (GFT) \cite{shuman2013emerging} of the signal $\mathbf{x}$ as follows
\begin{equation}
\hat{\mathbf{x}}=\Phi^H \mathbf{x}.
\label{eq:gft}
\end{equation}
\section{SDFT - 1D case}
The forward one dimensional discrete Fourier transform (1D-DFT) of the signal $\mathbf{x}\in\mathbb{R}^N$ can be computed in the following way
\[
\hat{\mathbf{x}}_k=\sum_{n=0}^{N-1} \mathbf{x}_n \e^{-i\frac{2\pi kn}{N}}.
\] 
We can write it in matrix form $\hat{\mathbf{x}}=V\mathbf{x}$,
where $V_{kn}=\e^{-i\frac{2\pi kn}{N}}=\rho_k^n$. $V\in\mathbb{C}^{N\times N}$ is the 1D-DFT matrix and it has the following property.
\begin{theorem}[Theorem 5.1 \cite{grady2010discrete}] 
\label{theo}
The rows of the DFT matrix are eigenvectors of any circulant matrix.
\end{theorem}
\begin{figure}
 \centering
   \includegraphics[width=2.4cm]{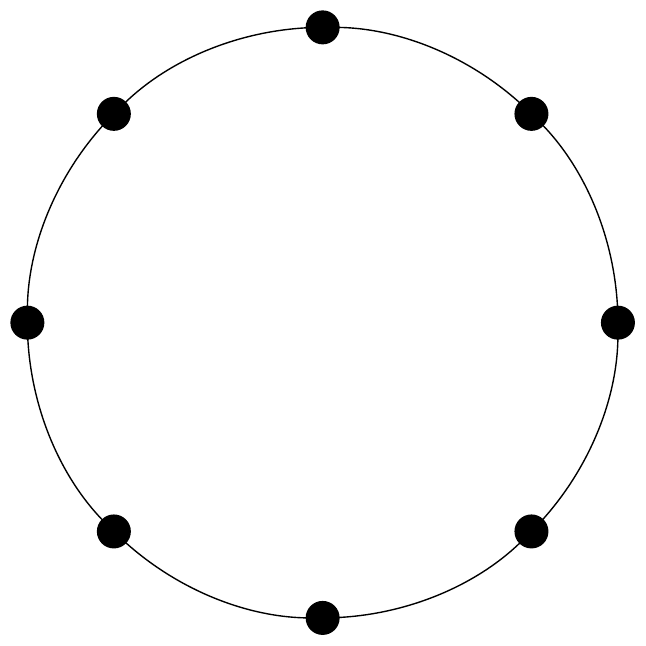}
   \hspace{0.5cm}
   \includegraphics[width=4.6cm]{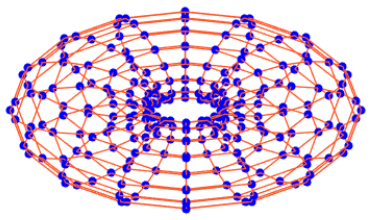}
 \caption{On the left: a cycle graph $\C_8$. On the right: A toroidal grid graph $\T_{16,16}$.}
 \label{fig:graph}
 \vspace{-0.5cm}
 \end{figure}
Let us now consider an undirected cycle graph $\C_N$ with $N$ vertices, whose structure is shown in Fig. \ref{fig:graph}. This type of graph is called a circulant graph because its adjacency matrix, and therefore its Laplacian matrix, is circulant. Circulant graphs are of great importance in graph signal processing, because they accommodate fundamental signal processing operations, such as linear shift-invariant filtering, downsampling, upsampling, and reconstruction \cite{ekambaram2013multiresolution,ekambaram2013circulant}. It is well known that a valid set of eigenvectors for any circulant matrix is the set of DFT matrix rows, then the 1D-DFT is a valid GFT for $\C_N$ (i.e. $\Phi^H=V$). However, repeated eigenvalues are present in the spectrum of $L(\C_N)$, because the following property holds 
\begin{equation}
\label{property}
\lambda_k=\lambda_{N-k},
\end{equation}
where $\lambda_k$ is the $k$-th eigenvalue of $L(\C_N)$ with $k=1,2,...,\frac{N}{2}-1$ \cite{tee2005eigenvectors}. The eigenvalues $\lambda_k$ can be computed in the following way \cite{tee2005eigenvectors}:
\begin{equation}
\label{circeig}
{\lambda_k=2-2\cos\frac{2\pi k}{N},}
\end{equation}
for $k=0,1,2,...,N-1$. In addition, $L(\C_N)$ has $N$ orthogonal eigenvectors $\{\mathbf{v}^{(k)}\}$, where $\mathbf{v}^{(k)}=\left[\rho_k, \rho_k^2,...,\rho_k^{n-1}\right]^T$ for $k=0,1,2,...,N-1$ \cite{tee2005eigenvectors}.

From \eqref{property} and \eqref{circeig}, we can state that, if $N$ is even, $\lambda_0$ and $\lambda_{\frac{N}{2}}$ have algebraic multiplicity 1, instead all the other eigenvalues have algebraic multiplicity 2 with $\lambda_k=\lambda_{N-k}$, where $1\le k \le \frac{N}{2}-1$. Since the eigenvectors are orthogonal, the geometric multiplicity is equal to the algebraic multiplicity. This means that the dimension of the eigenspaces corresponding to $\lambda_k$ where $1\le k \le \frac{N}{2}-1$ is 2, then the vector basis of the 1D-DFT is not the only possible eigenbasis of $L(\C_N)$.

We can then introduce the following corollary, whose proof follows from the discussion above and is omitted for brevity.
\begin{corollary} The graph Fourier transform of a cycle graph $\C_N$ may be equal to the 1D-DFT, but it is not the only possible graph Fourier transform of a cycle graph.
\end{corollary}

We now proceed to define the 1D-SDFT. Given an eigenvalue $\lambda_k$ of $L(\C_N)$ with multiplicity 2 and the two corresponding 1D-DFT vectors $\mathbf{v}^{(k)}$ and $\mathbf{v}^{(N-k)}$, we can define any other possible basis of the eigenspace corresponding to $\lambda_k$ as the result of a rotation of $\mathbf{v}^{(k)}$ and $\mathbf{v}^{(N-k)}$
\begin{equation}
\label{eq:rot}
\begin{bmatrix}
\mathbf{v}^{(k)'}\\
\mathbf{v}^{(N-k)'}
\end{bmatrix}
=
\begin{bmatrix}
\cos\theta_{k} & \sin\theta_{k}\\
-\sin\theta_{k} & \cos\theta_{k}
\end{bmatrix}
\begin{bmatrix}
\mathbf{v}^{(k)}\\
\mathbf{v}^{(N-k)}
\end{bmatrix},
\end{equation}
where $\theta_k$ is an angle in $[0,2\pi]$. 

For every $\lambda_k$ where $1\le k \le \frac{N}{2}-1$, we can rotate the corresponding eigenvectors as shown in \eqref{eq:rot}. In the 1D-DFT matrix, the pairs $\mathbf{v}^{(k)}$ and $\mathbf{v}^{(N-k)}$ are replaced with the rotated ones $\mathbf{v}'^{(k)}$ and $\mathbf{v}'^{(N-k)}$ obtaining a new transform matrix $V(\theta)\in\mathbb{C}^{N\times N}$ called 1D-SDFT. The vector $\theta\in\R^p$ contains all the rotation angles used and its length is $p=\frac{N}{2}-1$. The new transform matrix $V(\theta)$ can be written as 
\begin{equation}
\label{sdft-1d}
V(\theta)=R(\theta)V,
\end{equation}
where $V=V(0)\in\mathbb{C}^{N\times N}$ is the 1D-DFT matrix and $R(\theta)\in\R^{N\times N}$ is the rotation matrix, whose structure is defined so that, for each pair of eigenvectors, it performs the rotation as defined in \eqref{eq:rot}. It is important to underline that the choice of the eigenvector pairs is given by the analysis of the eigenvalue multiplicity. In this way, the transform defined in \eqref{sdft-1d} is still the graph transform of a cycle graph. 

Equation \eqref{sdft-1d} shows that the SDFT can be obtained by applying the rotation described by $R(\theta)$ to the output of the standard DFT, that can be easily computed using the FFT.

From a geometrical point of view, \eqref{eq:rot} represents a rotation in the complex plane. Given a real-valued signal $\mathbf{x}\in\R^N$, its DFT coefficients $\hat{\mathbf{x}}$ have the symmetry property 
$\hat{\mathbf{x}}_k=\hat{\mathbf{x}}_{N-k}^*$,
where $1\le k\le \frac{N}{2}-1$ and the ``$*$'' symbol denotes conjugation \cite{lyons2010understanding}. Then, using the rotation in \eqref{eq:rot} we can break this symmetry. For example, if we perform a rotation by $\frac{\pi}{4}$, we can completely separate the real part and the imaginary part. In fact, given $\mathbf{v}^{(k)'}$ and $\mathbf{v}^{(N-k)'}$, which are obtained rotating $\mathbf{v}^{(k)}$ and $\mathbf{v}^{(N-k)}$ by $\frac{\pi}{4}$ as in \eqref{eq:rot}, the new transform coefficients are
\[
\begin{split}
\begin{bmatrix}
\hat{\mathbf{x}}'_k\\
\hat{\mathbf{x}}'_{N-k}
\end{bmatrix}
&=
\begin{bmatrix}
\mathbf{v}^{(k)'}\\
\mathbf{v}^{(N-k)'}
\end{bmatrix}
\mathbf{x}
=
\begin{bmatrix}
\cos\frac{\pi}{4} & \sin\frac{\pi}{4}\\
-\sin\frac{\pi}{4} & \cos\frac{\pi}{4}
\end{bmatrix}
\begin{bmatrix}
\mathbf{v}^{(k)}\\
\mathbf{v}^{(N-k)}
\end{bmatrix}
\mathbf{x}\\
&=
\begin{bmatrix}
\frac{\sqrt{2}}{2} & \frac{\sqrt{2}}{2}\\
-\frac{\sqrt{2}}{2} & \frac{\sqrt{2}}{2}
\end{bmatrix}
\begin{bmatrix}
\hat{\mathbf{x}}_k\\
\hat{\mathbf{x}}_{N-k}
\end{bmatrix}
=
\begin{bmatrix}
\sqrt{2}\mbox{ Re}(\hat{\mathbf{x}}_k)\\
-i\sqrt{2}\mbox{ Im}(\hat{\mathbf{x}}_k)
\end{bmatrix},
\end{split}
\]
where $i=\sqrt{-1}$.
As an example, Fig. \ref{fig:plot1d} shows a plot of a pair of coefficients $\hat{\mathbf{x}}'_k$ and $\hat{\mathbf{x}}'_{N-k}$ as a function of the rotation angle $\theta_k\in [0,2\pi]$. For both coefficients, we can clearly see that the absolute values of the real and imaginary part are inversely proportional. Moreover, when $\theta_k=(2t+1)\frac{\pi}{4}$ with $t=0,1,2,3$ one coefficient is a real value and the other one is a pure imaginary value. 
\begin{figure}[t]
\centering
\includegraphics[width=7.5cm]{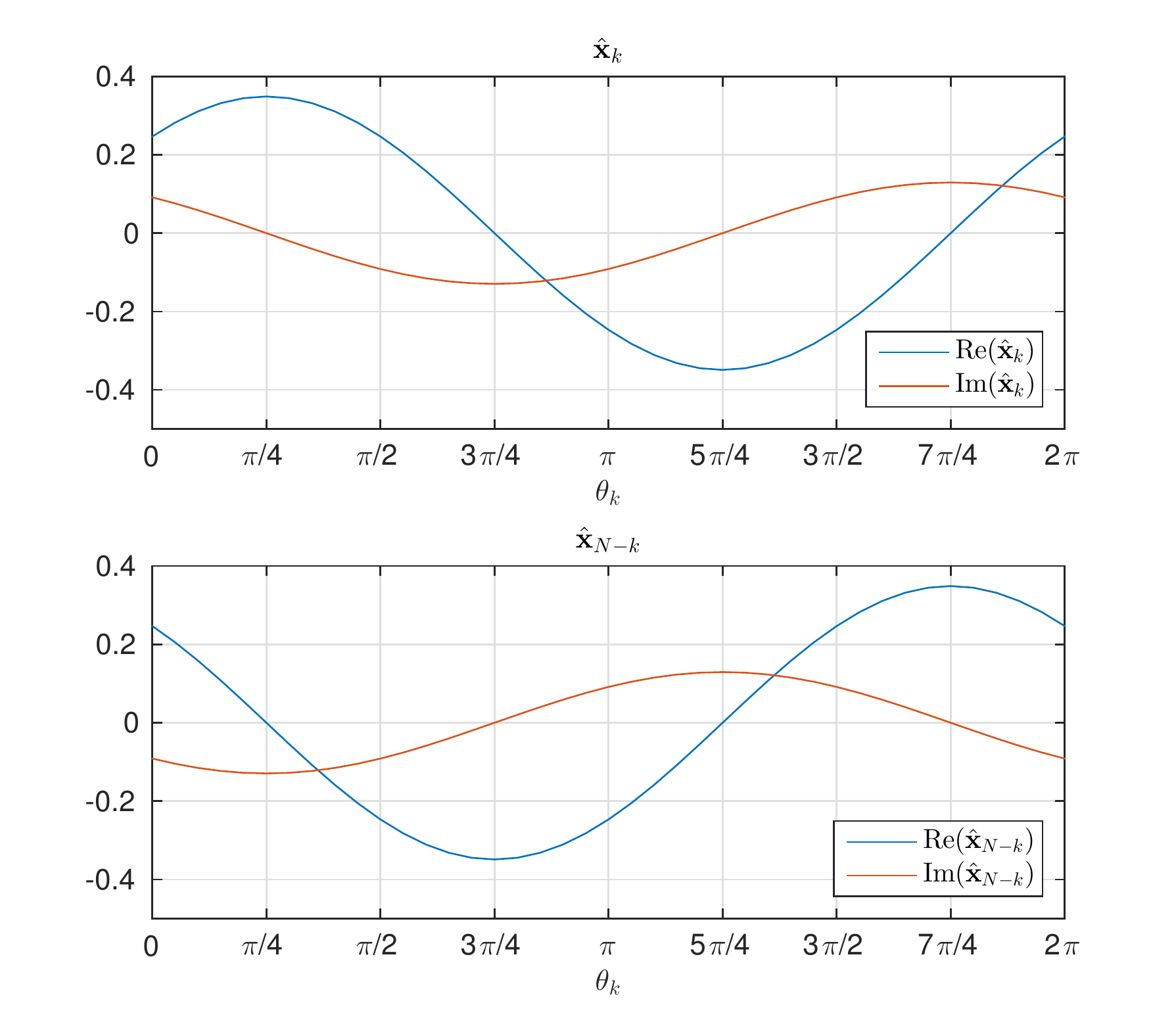}
\vspace{-0.6cm}
\caption{An example of a pair of coefficients $\hat{\mathbf{x}}'_k$ and $\hat{\mathbf{x}}'_{N-k}$ as a function of the rotation angle $\theta\in [0,2\pi]$.}
\label{fig:plot1d}
\vspace{-0.5cm}
\end{figure}
Finally, it is important to underline that the rotation described in \eqref{eq:rot} preserves the total energy of the coefficients in the eigenspace, i.e. 
$|\hat{\mathbf{x}}_k|^2+|\hat{\mathbf{x}}_{N-k}|^2=|\hat{\mathbf{x}}_k'|^2+|\hat{\mathbf{x}}_{N-k}'|^2.$

\subsection{Relationships of the 1D-SDFT to other transforms}
We have already shown that if $\theta=0$ the 1D-SDFT is equal to the 1D-DFT. Instead if $\theta=\frac{\pi}{4}$, it is interesting to show that for $1\le k\le\frac{N}{2}-1$ we have that $\hat{\mathbf{x}}_k=\sqrt{2}\mathbf{x}^{\mbox{cos}}_k$, where $\mathbf{x}^{\mbox{cos}}_k$ is the $k$-th coefficient of the Fourier cosine transform \cite{poularikas2010transforms}. Analogously, for $\frac{N}{2}+1\le k\le N-1$ we have that $\hat{\mathbf{x}}_k=-i\sqrt{2}\mathbf{x}^{\mbox{sin}}_k$, where $\mathbf{x}^{\mbox{sin}}_k$ is the $k$-th coefficient of the Fourier sine transform \cite{poularikas2010transforms}. In turn, the Fourier cosine and sine transform are highly related respectively to the DCT and DST \cite{poularikas2010transforms}.

Moreover, we can also relate the 1D-SDFT with the Hilbert transform \cite{kschischang2006hilbert}. In fact, given a signal $\mathbf{x}$ it can be proved that
\begin {equation}
\label{hilbert}
\mathcal{H}(\mathbf{x})=\mbox{ Im}\left(\tilde{V}\left(\frac{\pi}{4}\right)^HV\left(-\frac{\pi}{4}\right)\mathbf{x}\right),
\end{equation}
where $\mathcal{H}(\mathbf{x})$ is the Hilbert transform of $\mathbf{x}$ and $\tilde{V}\left(\frac{\pi}{4}\right)$ is the SDFT-1D correspondig to the improper rotation \cite{salomon2012computer}
\[
\begin{bmatrix}
\cos\frac{\pi}{4} & \sin\frac{\pi}{4}\\
\sin\frac{\pi}{4} & -\cos\frac{\pi}{4}
\end{bmatrix},
\]
moreover $\mbox{Re}\left(\tilde{V}\left(\frac{\pi}{4}\right)^HV\left(-\frac{\pi}{4}\right)\mathbf{x}\right)=(\mathbf{v}^{(0)^T}\mathbf{x})\mathbf{v}^{(0)}+(\mathbf{v}^{\left(\frac{N}{2}\right)^T}\mathbf{x})\mathbf{v}^{\left(\frac{N}{2}\right)}$.

These relationships are interesting because they may open the way to new generalizations of these transforms.

\section{SDFT - 2D case}
In the two dimensional case, the 2D-DFT of a signal $X\in\mathbb{R}^{N_1\times N_2}$ can be computed as follows
\[
\hat{X}_{kl}=\sum_{n=0}^{N_1-1}\sum_{m=0}^{N_2-1}X_{mn}\e^{-i2\pi\left(\frac{l}{N_1}n+\frac{k}{N_2}m\right)},
\]
in matrix form we can write it as
$\hat{\mathbf{x}}=W\mathbf{x}$,
where $\mathbf{x}\in\mathbb{R}^{N_1 N_2}$ is the vectorized signal $X$ and $W\in\mathbb{C}^{N_1 N_2\times N_1 N_2}$ is the 2D-DFT matrix, which is defined in the following way
\[
W_{ts}=\e^{-i2\pi\left(\frac{l}{N_1}n+\frac{k}{N_2}m\right)}=\rho_l^n\rho_k^m,
\]
where $s=mN_1+n$, $t=kN_1+l$, $0\le l,n\le N_1-1$ and $0\le m,k\le N_2-1$.

We now consider a grid graph with periodic boundary conditions that is called toroidal grid graph $\T_{N_1N_2}$ \cite{park2015many}, where $|\V|=N_1N_2$. An example of a toroidal grid graph is shown in Fig. \ref{fig:graph}. It is known that the toroidal grid graph $\mathcal{T}_{N_1N_2}$ corresponds to the product graph $\mathcal{C}_{N_1}\times\mathcal{C}_{N_2}$, where $\mathcal{C}_{N_i}$ is a cycle of $N_i$ vertices \cite{ruskey2003bent}.

In order to study the spectrum of the toroidal grid graph, we recall the following theorem on the spectrum of the product graph.
\begin{theorem}[Theorem 2.21 in \cite{merris1994laplacian}; \cite{merris1998laplacian}]\label{theo:compute_eigen}
 Let $\G_1$ and $\G_2$ be graphs on $N_1$ and $N_2$ vertices, respectively. Then, the eigenvalues of $L(\G_1\times \G_2)$ are all possible sums of $\lambda_i(\G_1)+\lambda_j(\G_2)$, where $0\le i\le N_1-1$ and $0\le j\le N_2-1$. Moreover, if $\mathbf{v}^{(i)}$ is an eigenvector of $\G_1$ corresponding to $\lambda_i(\G_1)$ and $\mathbf{v}^{(j)}$ an eigenvector of $\G_2$ corresponding to $\lambda_j(\G_2)$, then $\mathbf{v}^{(i)}\otimes \mathbf{v}^{(j)}$ (where $\otimes$ indicates the Kronecker product) is an eigenvector of $\G$ corresponding to $\lambda_i(\G_1)+\lambda_j(\G_2)$.
\end{theorem}

We now show that there is a strong connection between 2D-DFT and toroidal grid graph.
\begin{theorem}
Let $\mathcal{T}_{NN}$ be a toroidal graph, then the 2D-DFT basis is an eigenbasis of  $L(\mathcal{T}_{NN})$.
\end{theorem}
\begin{proof}
Let $\mathbf{v}^{(p)}$ and $\mathbf{v}^{(q)}$, where $0\le p,q\le N-1$, be the eigenvectors of $\mathcal{C}_{N}$ corrisponding respectively to the eigenvalues $\lambda_p$ and $\lambda_q$, as defined in Sec. III.   Then, using Theorem \ref{theo:compute_eigen} we can compute the eigenvector $\mathbf{u}^{(p,q)}$ of $\mathcal{T}_{NN}$ corresponding to the eigenvalue $\mu_{p,q}=\lambda_p+\lambda_q$
\[
\mathbf{u}^{(p,q)}=\mathbf{v}^{(p)}\otimes\mathbf{v}^{(q)}=\begin{bmatrix}\mathbf{v}^{(p)}_1\mathbf{v}^{(q)}\\
\mathbf{v}^{(p)}_2\mathbf{v}^{(q)}\\
\vdots\\
\mathbf{v}^{(p)}_{n-1}\mathbf{v}^{(q)}\\
\end{bmatrix}=\begin{bmatrix}\mathbf{v}^{(q)}\\
\rho_p\mathbf{v}^{(q)}\\
\vdots\\
\rho_p^{n-1}\mathbf{v}^{(q)}
\end{bmatrix}=\mathbf{w}^{(k)},
\]
where $k=pN+q$ and $\mathbf{w}^{(k)^T}$ is the $k$-th row of the 2D-DFT matrix $W$. Therefore, the 2D-DFT is an eigenbasis of the Laplacian of $\mathcal{T}_{NN}$ (i.e. $\Phi^H=W$).

\end{proof}
Since $\mu_{p,q}=\lambda_p+\lambda_q=\mu_{q,p}$ and recalling property \eqref{property} for the eigenvalues $\lambda_k$ of $\C_N$, in the spectrum of $L(\T_{NN})$ several repeated eigenvalues are presents:
\begin{itemize}
\item The eigenvalues $\mu_{p,q}$ where $1\le p,q\le \frac{N}{2}-1$ and $p\ne q$ have algebraic multiplicity 8 since $\mu_{p,q}=\mu_{q,p}=\mu_{p,N-q}=\mu_{N-q,p}=\mu_{N-p,q}=\mu_{q,N-p}=\mu_{N-p,N-q}=\mu_{N-q,N-p}$.
\item The eigenvalues $\mu_{p,p}$ where $1\le p\le \frac{N}{2}-1$ have algebraic multiplicity 4 since $\mu_{p,p}=\mu_{p,N-p}=\mu_{N-p,p}=\mu_{N-p,N-p}$. 
\item The eigenvalues $\mu_{p,q}$ where $p=0,\frac{N}{2}$ and $1\le q\le \frac{N}{2}-1$ (or $1\le p\le \frac{N}{2}-1$ and $q=0,\frac{N}{2}$) have algebraic multiplicity 4 because $\mu_{p,q}=\mu_{q,p}=\mu_{p,N-q}=\mu_{N-q,p}$ ($\mu_{p,q}=\mu_{q,p}=\mu_{N-p,q}=\mu_{q,N-p}$). 
\item The eigenvalue $\mu_{0,\frac{N}{2}}=\mu_{\frac{N}{2},0}$ has multiplicity 2.
\item The eigenvalues $\mu_{0,0}$ and $\mu_{\frac{N}{2},\frac{N}{2}}$ are the only ones with algebraic multiplicity 1.
\end{itemize}

Since the Kronecker product is not commutative, the eigenvectors $\mathbf{u}^{(p,q)}$ of $\T_{NN}$ are orthogonal. Then, the geometric multiplicity is equal to the algebraic multiplicity. Therefore, the dimension of the eigenspaces corresponding to the repeated eigenvalues is bigger than one. This proves that the 2D-DFT is not the unique eigenbasis for $L(\T_{NN})$ and, thus, the 2D-DFT is not the unique GFT for $\T_{NN}$.

As shown above, in the spectrum of $L(\T_{NN})$ many eigenvalues with multiplicity greater than 2 are present. Therefore, it may be possible to define rotations in more than two dimensions. However, these rotations may not have a clear geometrical meaning. For this reason, in the following of this section we restrict our study to rotations in two dimensions that exploit the symmetric property $\mu_{p,q}=\mu_{q,p}$. Instead, the rotations that exploit the property $\mu_{p,q}=\mu_{N-q,N-p}$ and $\mu_{N-p,q}=\mu_{p,N-q}$ are analog to the ones shown in the 1D case.

Given any vector pair of the 2D-DFT, $\mathbf{u}^{(p,q)}$ and  $\mathbf{u}^{(q,p)}$ where $p\ne q$, we can obtain a new pair of eigenvectors of $L(\T_{NN})$ by performing the following rotation
\begin{equation}
\label{rot2}
\begin{bmatrix}
\mathbf{u}^{(p,q)'}\\
\mathbf{u}^{(q,p)'}
\end{bmatrix}
=
\begin{bmatrix}
\cos\theta_{p,q} & \sin\theta_{p,q}\\
-\sin\theta_{p,q} & \cos\theta_{p,q}
\end{bmatrix}
\begin{bmatrix}
\mathbf{u}^{(p,q)}\\
\mathbf{u}^{(q,p)}
\end{bmatrix},
\end{equation}
where $\theta_{p,q}$ is an angle in $[0,2\pi]$.
Then, analogously to the 1D case, we can define a new transform matrix $V(\theta)\in\mathbb{C}^{N^2\times N^2}$, called 2D-SDFT, that is obtained by replacing in the 2D-DFT matrix the pairs $\mathbf{u}^{(p,q)}$ and  $\mathbf{u}^{(q,p)}$ with the rotated ones $\mathbf{u}^{(p,q)'}$ and  $\mathbf{u}^{(q,p)'}$. The vector $\theta\in\R^{p}$ contains all the angles used and its length is equal to the number of vector pairs, that is $p=\frac{N(N-1)}{2}$. Similarly to the 1D case, also the 2D-SDFT matrix $V(\theta)$ can be computed as in \eqref{sdft-1d}, where, in this case, $R(\theta)\in\R^{N^2\times N^2}$ is the rotation matrix whose structure is defined so that, for each pair of vectors, it performs the rotation as defined in \eqref{rot2}.

Given a signal $\mathbf{x}\in\R^{N\times N}$, we can compute the SDFT coefficients of $\mathbf{x}$ corresponding to the eigenvectors $\mathbf{u}^{(p,q)'}$ and $\mathbf{u}^{(q,p)'}$ in the following way
\[
\begin{split}
\begin{bmatrix}
\hat{\mathbf{x}}'_{p,q}\\
\hat{\mathbf{x}}'_{q,p}
\end{bmatrix}
&=
\begin{bmatrix}
\mathbf{u}^{(p,q)'}\\
\mathbf{u}^{(q,p)'}
\end{bmatrix}
\mathbf{x}
=
\begin{bmatrix}
\cos\theta_{p,q} & \sin\theta_{p,q}\\
-\sin\theta_{p,q} & \cos\theta_{p,q}
\end{bmatrix}
\begin{bmatrix}
\mathbf{u}^{(p,q)}\\
\mathbf{u}^{(q,p)}
\end{bmatrix}
\mathbf{x}\\
&=
\begin{bmatrix}
\cos\theta_{p,q} & \sin\theta_{p,q}\\
-\sin\theta_{p,q} & \cos\theta_{p,q}
\end{bmatrix}
\begin{bmatrix}
\hat{\mathbf{x}}_{p,q}\\
\hat{\mathbf{x}}_{q,p}
\end{bmatrix}\\
&=\begin{bmatrix}
\cos\theta_{p,q} & \sin\theta_{p,q}\\
-\sin\theta_{p,q} & \cos\theta_{p,q}
\end{bmatrix}
\begin{bmatrix}
\mbox{Re}(\hat{\mathbf{x}}_{p,q})\\
\mbox{Re}(\hat{\mathbf{x}}_{p,q})
\end{bmatrix}+\\
&+i\begin{bmatrix}
\cos\theta_{p,q} & \sin\theta_{p,q}\\
-\sin\theta_{p,q} & \cos\theta_{p,q}
\end{bmatrix}
\begin{bmatrix}
\mbox{Im}(\hat{\mathbf{x}}_{p,q})\\
\mbox{Im}(\hat{\mathbf{x}}_{p,q})
\end{bmatrix}.
\end{split}
\]
Therefore, we can state that, from a geometrical point of view, \eqref{rot2} performs separately a rotation of the real and imaginary part in the 2D Euclidean space.  
Then, by applying \eqref{rot2} the total energy of the real and imaginary part of the coefficient pair remains unchanged, that is
\[
\mbox{Re}(\hat{\mathbf{x}}_{p,q})^2+\mbox{Re}(\hat{\mathbf{x}}_{q,p})^2=\mbox{Re}(\hat{\mathbf{x}}'_{p,q})^2+\mbox{Re}(\hat{\mathbf{x}}'_{q,p})^2,
\]
\[
\mbox{Im}(\hat{\mathbf{x}}_{p,q})^2+\mbox{Im}(\hat{\mathbf{x}}_{q,p})^2=\mbox{Im}(\hat{\mathbf{x}}'_{p,q})^2+\mbox{Im}(\hat{\mathbf{x}}'_{q,p})^2,
\]
but it is possible to unbalance the energy of the real and imaginary part of each coefficient. 
For example, we can compact all the energy of the real part in one coefficient, zeroing out the other one. In fact, given the pair of DFT coefficients $\hat{\mathbf{x}}_{p,q}$ and $\hat{\mathbf{x}}_{q,p}$ we rotate the pair of corresponding DFT vectors $\mathbf{u}^{(p,q)}$ and $\mathbf{u}^{(q,p)}$ as in \eqref{rot2} by an angle defined as follows
\[
\theta_{p,q}=\mbox{arctan }\frac{\mbox{Re}(\hat{\mathbf{x}}_{q,p})}{\mbox{Re}(\hat{\mathbf{x}}_{p,q})}.
\]
Then, we get that $\mbox{Re}(\hat{\mathbf{x}}_{q,p}')=0$ and the energy of the real part of the coefficient pair is conveyed to $\hat{\mathbf{x}}_{p,q}'$, as shown in Fig. \ref{fig:rot_xy}.

\begin{figure}[t]
\centering
\includegraphics[width=3.8cm]{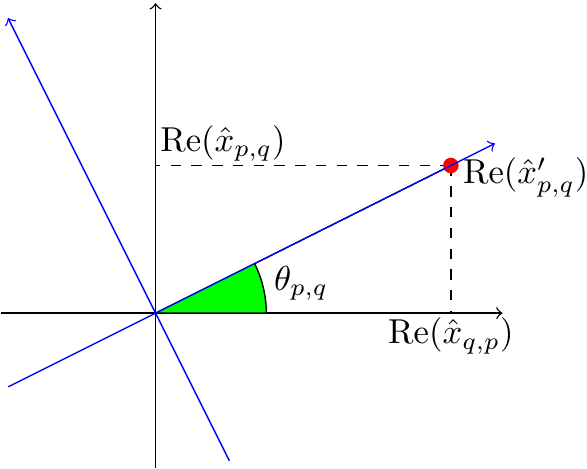}
\vspace{-0.2cm}
\caption{Rotation of the 2D-DFT vector pair $\mathbf{u}^{(p,q)}$ and $\mathbf{u}^{(q,p)}$.}
\label{fig:rot_xy}
\vspace{-0.5cm}
\end{figure}
\vspace{-0.3cm}
\section{Applications of the SDFT}
In this section we discuss possible applications of the SDFT. 

The 1D-SDFT can be useful for signal analysis and processing. For example, it can be used for easily filtering the even/odd component of a signal. In fact, if we rotate the pairs of vectors $\mathbf{v}^{(k)}$ and $\mathbf{v}^{(N-k)}$ by $\frac{\pi}{4}$, we can design a filter that, convolved with the input signal, retains only the first (last) $\frac{N}{2}$ coefficients and outputs the even (odd) signal component, as shown in Fig. \ref{fig:sym} where we obtain as output of the filter the even component of the input signal. We can also easily filter the even or odd component of specific frequencies.  Analogously in 2D, we can perform the same filtering operation by rotating by $\frac{\pi}{4}$ the pairs of vectors $\mathbf{u}^{(p,q)}$ and $\mathbf{u}^{(N-p,N-q)}$ and the pairs $\mathbf{u}^{(p,N-q)}$ and $\mathbf{u}^{(N-p,q)}$. This filtering operation could be useful for signal representation, as in \cite{gnutti2015representation}.
\begin{figure}[t]
\centering
\includegraphics[width=7.3cm]{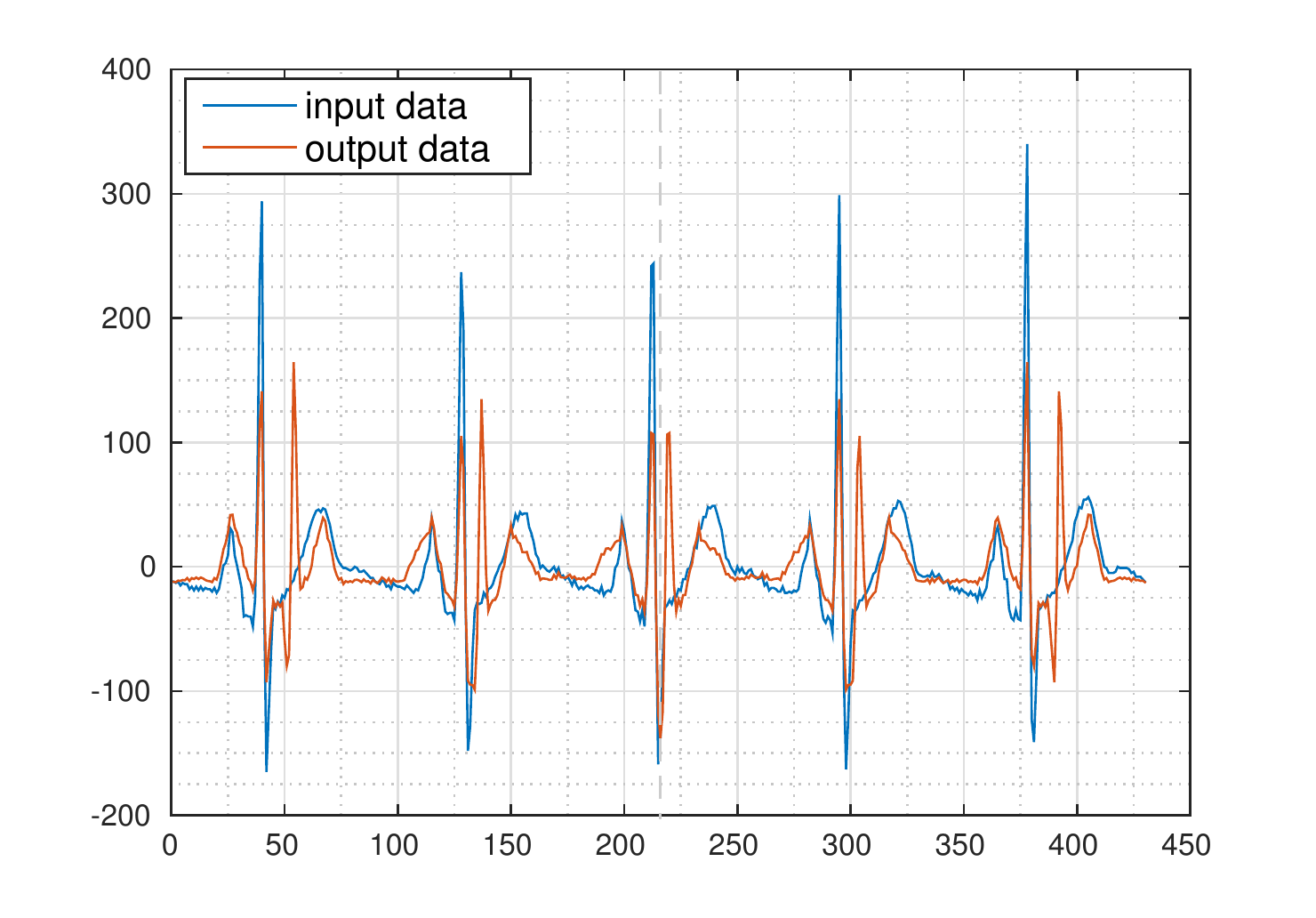}
\vspace{-0.5cm}
\caption{Example of filtering the even component of a real signal.}
\label{fig:sym}
\vspace{-0.5cm}
\end{figure}

Moreover, in \eqref{hilbert} we have already shown that the SDFT-1D may be used to perform the Hilbert transform, this could be useful for computing the local phase and amplitude, that is used in many applications, such as edge detection \cite{kovesi1999image} and image feature extraction \cite{carneiro2002phase}.

The 1D-SDFT can be applied also in multimedia encryption problems. In this field, several works use parametrized versions of common transforms for security purposes \cite{pande2013secure, unnikrishnan2000optical}. Since the SDFT is a parametrized version of the DFT, one can use the parameter $\theta\in\mathbb{R}^p$ as a secret key. More specifically, given a signal $\mathbf{x}\in\mathbb{R}^p$ we can obtain $\hat{\mathbf{x}}=V(\theta)\mathbf{x}$ and then consider the first $\frac{N}{2}$ components of $\hat{\mathbf{x}}$ as the encrypted signal. Given $\theta$, it is possible to reconstruct $\hat{\mathbf{x}}$ and then we can obtain the original signal $\mathbf{x}$ by applying the inverse SDFT.  We can also consider the SDFT as a keyed transform basis that can be used for compressed sensing-based cryptography \cite{bianchi2016analysis,zhang2014bi}.

The applications presented in this section are just a few examples of possible applications of the SDFT, but the SDFT could be of interest for a wide range of fields, such as array signal processing, phase retrieval and magnetic resonance imaging.
\vspace{-0.5cm}
\section{Conclusion}
The proposed SDFT is an important generalization of the classical DFT and it may be of interest in a wide range of application fields, such as filtering, signal analysis and multimedia encryption.


\begin{thebibliography}{10}
\providecommand{\url}[1]{#1}
\csname url@samestyle\endcsname
\providecommand{\newblock}{\relax}
\providecommand{\bibinfo}[2]{#2}
\providecommand{\BIBentrySTDinterwordspacing}{\spaceskip=0pt\relax}
\providecommand{\BIBentryALTinterwordstretchfactor}{4}
\providecommand{\BIBentryALTinterwordspacing}{\spaceskip=\fontdimen2\font plus
\BIBentryALTinterwordstretchfactor\fontdimen3\font minus
  \fontdimen4\font\relax}
\providecommand{\BIBforeignlanguage}[2]{{%
\expandafter\ifx\csname l@#1\endcsname\relax
\typeout{** WARNING: IEEEtran.bst: No hyphenation pattern has been}%
\typeout{** loaded for the language `#1'. Using the pattern for}%
\typeout{** the default language instead.}%
\else
\language=\csname l@#1\endcsname
\fi
#2}}
\providecommand{\BIBdecl}{\relax}
\BIBdecl

\bibitem{zeng2008directional}
B.~Zeng and J.~Fu, ``Directional discrete cosine transforms - a new framework
  for image coding,'' \emph{IEEE Trans. Circuits Syst. Video Technol.},
  vol.~18, no.~3, pp. 305--313, 2008.

\bibitem{fracastoro2015steerable}
G.~Fracastoro and E.~Magli, ``Steerable discrete cosine transform,'' in
  \emph{Proc. IEEE International Workshop on Multimedia Signal Processing, 2015
  (MMSP)}.\hskip 1em plus 0.5em minus 0.4em\relax IEEE, 2015, pp. 1--6.

\bibitem{alshina2011rotational}
E.~Alshina, A.~Alshin, and F.~C. Fernandes, ``Rotational transform for image
  and video compression,'' in \emph{Proc. IEEE International Conference on
  Image Processing}, 2011, pp. 3689--3692.

\bibitem{candes2000curvelets}
E.~Candes and D.~Donoho, ``Curvelets: A surprisingly effective nonadaptive
  representation for objects with edges,'' DTIC Document, Tech. Rep., 2000.

\bibitem{le2005sparse}
E.~L. Pennec and S.~Mallat, ``Sparse geometric image representations with
  bandelets,'' \emph{Image Processing, IEEE Transactions on}, vol.~14, no.~4,
  pp. 423--438, 2005.

\bibitem{do2005contourlet}
M.~N. Do and M.~Vetterli, ``The contourlet transform: an efficient directional
  multiresolution image representation,'' \emph{IEEE Transactions on image
  processing}, vol.~14, no.~12, pp. 2091--2106, 2005.

\bibitem{lyons2010understanding}
R.~G. Lyons, \emph{Understanding digital signal processing}.\hskip 1em plus
  0.5em minus 0.4em\relax Pearson Education, 2004.

\bibitem{allen1977unified}
J.~B. Allen and L.~R. Rabiner, ``A unified approach to short-time {F}ourier
  analysis and synthesis,'' \emph{Proceedings of the IEEE}, vol.~65, no.~11,
  pp. 1558--1564, 1977.

\bibitem{almeida1994fractional}
L.~B. Almeida, ``The fractional fourier transform and time-frequency
  representations,'' \emph{IEEE Transactions on signal processing}, vol.~42,
  no.~11, pp. 3084--3091, 1994.

\bibitem{almeida1993introduction}
------, ``An introduction to the angular {F}ourier transform,'' in \emph{IEEE
  International Conference on Acoustics, Speech, and Signal Processing
  (ICASSP)}, vol.~3.\hskip 1em plus 0.5em minus 0.4em\relax IEEE, 1993, pp.
  257--260.

\bibitem{santhanam1996discrete}
B.~Santhanam and J.~H. McClellan, ``The discrete rotational {F}ourier
  transform,'' \emph{IEEE Transactions on Signal Processing}, vol.~44, no.~4,
  pp. 994--998, 1996.

\bibitem{hammond2011wavelets}
D.~K. Hammond, P.~Vandergheynst, and R.~Gribonval, ``Wavelets on graphs via
  spectral graph theory,'' \emph{Applied and Computational Harmonic Analysis},
  vol.~30, no.~2, pp. 129--150, 2011.

\bibitem{shuman2013emerging}
D.~Shuman, S.~Narang, P.~Frossard, A.~Ortega, and P.~Vandergheynst, ``The
  emerging field of signal processing on graphs: extending high-dimensional
  data analysis to networks and other irregular domains,'' \emph{Signal
  Processing Magazine, IEEE}, vol.~30, no.~3, pp. 83--98, 2013.

\bibitem{strang1999discrete}
G.~Strang, ``The discrete cosine transform,'' \emph{SIAM review}, vol.~41,
  no.~1, pp. 135--147, 1999.

\bibitem{pande2013secure}
A.~Pande and J.~Zambreno, ``The secure wavelet transform,'' in \emph{Embedded
  Multimedia Security Systems}.\hskip 1em plus 0.5em minus 0.4em\relax
  Springer, 2013, pp. 67--89.

\bibitem{unnikrishnan2000optical}
G.~Unnikrishnan, J.~Joseph, and K.~Singh, ``Optical encryption by double-random
  phase encoding in the fractional fourier domain,'' \emph{Optics letters},
  vol.~25, no.~12, pp. 887--889, 2000.

\bibitem{grady2010discrete}
L.~J. Grady and J.~Polimeni, \emph{Discrete calculus: Applied analysis on
  graphs for computational science}.\hskip 1em plus 0.5em minus 0.4em\relax
  Springer Science \& Business Media, 2010.

\bibitem{ekambaram2013multiresolution}
V.~N. Ekambaram, G.~C. Fanti, B.~Ayazifar, and K.~Ramchandran,
  ``Multiresolution graph signal processing via circulant structures,'' in
  \emph{Proc. IEEE Digital Signal Processing and Signal Processing Education
  Meeting (DSP/SPE)}, 2013, pp. 112--117.

\bibitem{ekambaram2013circulant}
------, ``Circulant structures and graph signal processing,'' in \emph{Proc.
  IEEE International Conference on Image Processing (ICIP)}, 2013, pp.
  834--838.

\bibitem{tee2005eigenvectors}
G.~J. Tee, ``Eigenvectors of block circulant and alternating circulant
  matrices,'' \emph{New Zealand Journal of Mathematics}, vol.~36, pp. 195--211,
  2007.

\bibitem{poularikas2010transforms}
A.~D. Poularikas, \emph{Transforms and applications handbook}.\hskip 1em plus
  0.5em minus 0.4em\relax CRC press, 2010.

\bibitem{kschischang2006hilbert}
F.~R. Kschischang, ``The hilbert transform,'' \emph{University of Toronto},
  2006.

\bibitem{salomon2012computer}
D.~Salomon, \emph{Computer graphics and geometric modeling}.\hskip 1em plus
  0.5em minus 0.4em\relax Springer Science \& Business Media, 2012.

\bibitem{park2015many}
J.~H. Park and I.~Ihm, ``Many-to-many two-disjoint path covers in cylindrical
  and toroidal grids,'' \emph{Discrete Applied Mathematics}, vol. 185, pp.
  168--191, 2015.

\bibitem{ruskey2003bent}
F.~Ruskey and J.~Sawada, ``Bent hamilton cycles in d-dimensional grid graphs,''
  \emph{Electronic Journal of Combinatorics}, vol.~10, no. 1 R, 2003.

\bibitem{merris1994laplacian}
R.~Merris, ``Laplacian matrices of graphs: a survey,'' \emph{Linear algebra and
  its applications}, vol. 197, pp. 143--176, 1994.

\bibitem{merris1998laplacian}
------, ``Laplacian graph eigenvectors,'' \emph{Linear algebra and its
  applications}, vol. 278, no.~1, pp. 221--236, 1998.

\bibitem{gnutti2015representation}
A.~Gnutti, F.~Guerrini, and R.~Leonardi, ``Representation of signals by local
  symmetry decomposition,'' in \emph{European Signal Processing Conference
  (EUSIPCO)}, 2015, pp. 983--987.

\bibitem{kovesi1999image}
P.~Kovesi, ``Image features from phase congruency,'' \emph{Videre: Journal of
  computer vision research}, vol.~1, no.~3, pp. 1--26, 1999.

\bibitem{carneiro2002phase}
G.~Carneiro and A.~D. Jepson, ``Phase-based local features,'' in \emph{European
  Conference on Computer Vision}.\hskip 1em plus 0.5em minus 0.4em\relax
  Springer, 2002, pp. 282--296.

\bibitem{bianchi2016analysis}
T.~Bianchi, V.~Bioglio, and E.~Magli, ``Analysis of one-time random projections
  for privacy preserving compressed sensing,'' \emph{IEEE Transactions on
  Information Forensics and Security}, vol.~11, no.~2, pp. 313--327, 2016.

\bibitem{zhang2014bi}
L.~Y. Zhang, K.-W. Wong, Y.~Zhang, and J.~Zhou, ``Bi-level protected
  compressive sampling,'' \emph{arXiv preprint arXiv:1406.1725}, 2014.

\end{thebibliography}
\end{document}